\documentclass[11pt]{article}
\usepackage[latin1]{inputenc}
\usepackage{latexsym}

\usepackage[pdftex]{graphics}
\usepackage{amsmath, verbatim, epsfig, fullpage, amssymb,amsthm}

\newtheorem{thm}{Theorem}

\newtheorem{defi}[thm]{Definition}

\newtheorem{prop}[thm]{Proposition}

\begin{document}
\large
\renewcommand{\baselinestretch}{1}
\normalsize
\renewcommand{\baselinestretch}{1}
\parindent .5cm

\renewcommand{\baselinestretch}{2}
\normalsize
\renewcommand{\baselinestretch}{2}

\markright{
     Schoenberg, Hoffmann, and Harrigan. A recursive point process model for infectious diseases.
     }

\begin{center} {\bf A recursive point process model for infectious diseases.}\\[.5in]
\end{center}

\begin{center}
Schoenberg, Frederic P.$^1$

Hoffmann, Marc.$^2$ 

Harrigan, Ryan.$^3$

\end{center}

$^1$ Department of Statistics,
University of California, Los Angeles, CA 90095--1554, USA.\\

\hspace{1.65in} phone:  310-794-5193

\hspace{1.65in} fax: 310-206-5658

\hspace{1.65in} email: frederic@stat.ucla.edu

\hspace{1.65in} Postal address:  UCLA Dept.\ of Statistics

\hspace{2.8in} 8142 Math-Science Building

\hspace{2.8in} Los Angeles, CA 90095--1554, USA.

$^2$ Universit\'e Paris-Dauphine, PSL Research University, CNRS, Ceremade, 75016 Paris, France

$^3$ Institute of the Environment and Sustainability,
University of California, Los Angeles, CA 90095-1554, USA.\\

\pagebreak

{\bf Abstract.} We introduce a new type of point process model to describe the incidence of contagious diseases.
The model is a variant of the Hawkes self-exciting process and exhibits similar clustering but without the restriction that the component describing the contagion must remain static over time.
Instead, our proposed model prescribes that the degree of contagion (or productivity) changes as a function of the conditional intensity; of particular interest is the special case where the productivity is inversely proportional to the conditional intensity.  
The model incorporates the premise that when the disease occurs at very low frequency in the population, such as in the primary stages of an outbreak, then anyone with the disease is likely to have a high rate of transmission to others, whereas when the disease is prevalent in the population, then the transmission rate is lower due to human mitigation actions and prevention measures and a relatively high percentage of previous exposure in the total population.
The model is said to be recursive, in the sense that the conditional intensity at any particular time depends on the productivity associated with previous points, and this productivity in turn depends on the conditional intensity at those points. 
Some basic properties of the model are derived, estimation and simulation are discussed, and the recursive model is shown to fit well to historic data on measles in Los Angeles, California, 
a relevant example given the 2017 outbreak of this disease in the same region. \\

\section{Introduction.}

{\sl Hawkes} self-exciting point processes (Hawkes 1971) are 
a type of branching point process model that has become very commonly used in modeling clustered phenomena. 
For example, versions of Hawkes models are used
to model seismicity (Ogata 1988, 1998), crimes (Mohler et al.\ 2011), 
invasive plants (Balderama et al.\ 2012), 
terrorist strikes (Porter and White 2012), and perturbations in financial markets (Bacry et al.\ 2013 and Bacry et al.\ 2015).\\

Although Hawkes models have some features making them amenable to modeling incidence of infectious diseases,
consideration of the nature of the spread of disease may suggest a somewhat different type of model.
For instance, Hawkes processes have the property that the productivity (the expected number of secondary events triggered directly by the given event, or in the case of infectious disease, the expected number of transmissions from one individual to another) is static. 
In the case of Hawkes models applied to earthquakes (e.g.\ Ogata 1988, Ogata 1998), the basic Hawkes model was extended to allow the productivity of an earthquake to depend on its magnitude, but still not to depend on the time or location of the event, nor on the number of previously occurring events. 
When considering infectious diseases, however, this assumption of static productivity seems questionable. Early in the onset of an epidemic, when prevalence of the disease is still low, one would expect the rate of transmission to be much higher than when the prevalence of the disease is higher, because of human efforts at containment and intervention of the disease, and because some potential hosts of the disease may have already been exposed. Thus, we introduce a new type of point process model where the productivity (expected number of transmissions) for a subject infected at location $(s,t)$ in space-time is a function of the conditional intensity at $(s,t)$. Since the conditional intensity in turn depends critically on this productivity, we call the model {\sl recursive}.\\

Here we present this extension of Hawkes point process models as they apply to infectious diseases in the following format. After a brief review of point processes in general and Hawkes models in particular in Section 2, we introduce the recursive model in Section 3, followed by the derivation of some basic properties of the model in Section 4. Simulation and estimation are discussed in Sections 5 and 6, respectively, and in Section 7 we fit the model to data on recorded cases of measles in Los Angeles, California from 1910 to 1952.  Section 8 contains some concluding remarks.\\

\section{Hawkes point processes.}

A point process (Daley and Vere-Jones, 2003; Daley and Vere-Jones, 2007) is a $\sigma$-finite collection of points $\{\tau_1, \tau_2, ...\}$ 
occurring in some metric space $S$.
While the definitions and results below can be extended quite readily to other spaces, 
we will assume for simplicity throughout that the metric space $S$ is  
a bounded region $B \times [0,T]$ in space-time, with $\mu$ representing Lebesgue measure, 
and we will assume for convenience that the spatial region is scaled so that 
$\mu(B) = 1$. 
With this convention, all our definitions and results apply also to the case of a purely temporal process; 
in such cases one must simply ignore the integral over $B$ in formulae below.
A point process is {\sl simple} if, with probability 1, none of the points overlap exactly. \\

Spatial-temporal point processes are typically modeled via their conditional intensity, 
$\lambda(t)$ or $\lambda(s,t)$, which represents the infinitesimal rate at which points are 
accumulating at location $(s,t)$ of space-time, 
given information on all points occurring prior to time $t$.
Simple spatial-temporal point processes are uniquely characterized by their conditional intensity (Daley and Vere-Jones, 2007); 
for models for non-simple point processes, see Schoenberg (2006). 

For a simple spatial-temporal Hawkes process (Hawkes 1971), the conditional rate  
of events at location $(s,t)$ of space-time, given information ${\mathcal H}_t$ on all events prior to time $t$, can be written
\begin{eqnarray} \label{hawkes}
\lambda (s,t | \mathcal{H}_t) = \mu + K \int \limits _B \int \limits _0 ^t g(s-s',t-t') dN(s',t'),
\end{eqnarray}
where $\mu>0$, is the background rate, $g(v) \geq 0$ is the 
{\sl triggering density} satisfying $\int \limits _B \int_0^\infty g(u,v) du dv = 1$ which 
describes the spatial-temporal conductivity of events,
and the constant $K$ is the productivity, which is typically required to satisfy $0 \leq K < 1$ in order to ensure
stationarity and subcriticality (Hawkes, 1971). \\

Ogata (1988) extended the Hawkes model in order for earthquakes of different magnitudes to have different productivity.
Hawkes models and their extension to the temporal-magnitude case were called {\sl epidemic} by Ogata (1988), 
since they posit that an 
earthquake can produce aftershocks which in turn produce their own aftershocks, etc.
Several forms of the triggering function $g$ have been posited for describing seismological data,
such as 
$g(v) = \frac{1}{(v+c)^{p}},$
where $u$ is the time elapsed since a previous event (Ogata 1988).\\

Hawkes processes have been extended to describe the
space-time-magnitude distribution of seismic events. 
A version suggested by Ogata (1998)
uses a spatially inhomogeneous background rate and circular aftershock regions where the squared distance between
an aftershock and its triggering event follows a Pareto
distribution. 
The model may be written
\begin{equation*}
\lambda(s,t | \mathcal{H}_t)  = \mu(s) + 
K \int \limits _B \int \limits _0 ^t g(s-s', t-t', m') dN(s',t',m'),
\end{equation*}
with triggering function
\begin{equation}\textstyle
g(u,v,m)  =  (||u||^2 + d)^{-q} \exp\{a(m-M_0)\}(v+c)^{-p}, \label{ETAS} 
\end{equation}
where $||s_i - s_j||^2$ represents the squared distance between the epicenters or hypocenters 
${\bf s_i}$ and ${\bf s_j}$ of earthquakes $i$ and $j$, respectively, and 
$d>0$ and $q>0$ are parameters describing the spatial distribution of
triggered earthquakes about their respective mainshocks. \\

The ETAS model has been extended by allowing the parameters to vary spatially and temporally. 
For example, the HIST-ETAS model (Ogata et al.\ 2003, Ogata et al.\ 2004) assumes the parameters in the ETAS model are locally constant within small spatial-temporal cells. Similarly, 
Harte (2014) allows the ETAS model's productivity parameter to vary smoothly in space and time. 
In the following section we extend the model in a different way, allowing the productivity to vary as a function of $\lambda$. \\


\section{Proposed recursive model.}

Consideration of the nature of disease epidemics may lead one to question the usual assumption in Hawkes models of static productivity. 
For instance, when the prevalence of the disease is low or zero in a region, as is the case when the epidemic has never struck before or has not struck in considerable time, then the conditional intensity $\lambda$ is small and one would expect the rate of transmission for each infected person to be quite high, as a majority of hosts are likely immunologically naive, and a carrier of the disease may be expected to infect many others. When the epidemic is at its peak and many subjects have contracted the disease, on the other hand, $\lambda$ is large and 
one might expect the rate of transmission to be lower due to human efforts at containment and intervention of the disease, and because many subjects may have already been exposed and thus might be recovered and immune to further infection, or deceased (in either case no longer part of a susceptible pool). 
These considerations suggest a point process model where the productivity for a subject infected at location $(s,t)$ in space-time is inversely related to the conditional intensity at $(s,t)$. 
Since the conditional intensity in turn depends critically on this productivity, we call the model {\sl recursive}.

We may write this model
\begin{eqnarray}
\lambda(s,t) = \mu + \int \limits _B \int \limits _0 ^t H(\lambda_{s',t'}) \, g(s-s',t-t') dN(s',t'), \label{rec}
\end{eqnarray}
where $\mu > 0$, and $g>0$ is a density function. 
The productivity function $H$ should typically be decreasing in light of the considerations above regarding the transmission of disease,
and we focus in particular in what follows on the case where $H(x) = \kappa x^{- \alpha}$, 
with $\kappa > 0$, so that 
\begin{eqnarray}
\lambda(s,t) = \mu+ \kappa \int \limits _B \int \limits _0 ^t \lambda_{s',t'}^{-\alpha} \, g(s-s',t-t') dN(s',t'). \label{rec2}
\end{eqnarray}
The triggering density $g$ may be given e.g.\ by an exponential density,
\begin{eqnarray}
g(u,v) = \beta \exp(-\beta v), \label{exponential}
\end{eqnarray}
or exponential in space and time, 
\begin{eqnarray*}
g(u,v) = \beta_s \beta_t \exp(-\beta_s u - \beta_t v).
\end{eqnarray*}

When $\alpha = 0$, (\ref{rec2}) reduces to a Hawkes process.
We will refer to the special case where $\alpha=1$, i.e.\ where
\begin{eqnarray}
\lambda(s,t) = \mu+ \kappa \int \limits _B \int \limits _0 ^t \frac{g(s-s',t-t')}{\lambda_{s',t'}} \, dN(s',t') \label{oneover}
\end{eqnarray}
as {\sl standard}. 
The standard recursive model has especially simple and attractive features, some of which are described in Section 4.\\

\section{Basic properties of the recursion model.}

We prove the existence of a simple point process with conditional intensity \eqref{rec2}, and find the mean, variance, and certain large sample properties of the process. \\

\noindent {\bf Existence.} \\ 

\begin{prop} \label{existence}
Given a complete probability space, a recursive model with conditional intensity satisfying \eqref{rec2} can be constructed with $H(x)=\kappa x^{-\alpha}$, for any $\alpha, \kappa >0$.
\end{prop}
\begin{proof} Let $(e_k)_{k \geq 1}$ be a sequence of independent random variables. Set $T_0=0$ and
$$T_{k+1} = \inf\big\{t>T_k,\;\;\int \limits_{T_k}^{t-}\big(\mu+\kappa\sum_{i=1}^k\lambda^{-\alpha}_{T_i} g(s-T_i)\big)ds = e_{k+1}\big\}.$$
Define, for $k \geq 1$ the sequence of processes
$N_t^{(k)} = \sum_{i=1}^k {\bf 1}_{\{T_i \leq t\}}$. It is easy to see that $N_t^{(k)}$ is a counting process with stochastic intensity $\lambda_t^{(k)}$ satisfying
$$\lambda_t^{(k)} = \mu+\kappa\int \limits_0^{t-}(\lambda_u^{(k)})^{-\alpha}g(t-u)dN_u^{(k)}.$$
Let $N_t=\lim_{k \rightarrow \infty} N_t^{(k)}$. Let us show that $N_t$ is well-defined, {\it i.e.} has no accumulation of jumps. We have
\begin{align*}
E[N_{t}^{(k)}] & = E\big[\int \limits_0^t \lambda_{s-}^{(k)}ds\big] \\
& \leq \mu t + \kappa E\big[\int_0^t \int \limits _0 ^{{(s-)}} \big(\lambda_{u}^{(k)}\big)^{-\alpha} \, g(s-u) dN_{u}ds\big] \\
& \leq \mu t + \kappa \mu^{-\alpha} E\big[\int \limits_0^t \int \limits _0 ^{{(s-)}} g(s-u) dN_{u}^{(k)}ds\big] \\
& = \mu t + \kappa \mu^{-\alpha}E\big[\int_0^t g(t-s) N_{s}^{(k)}ds\big] 
 \end{align*}
where the last line can be obtained for instance by Lemma 22 in Delattre et al.\ 2016. Hence
$$E[N_{t}^{(k)}] \leq  \mu t + \kappa \mu^{-\alpha}\int_0^t g(t-s) E\big[N_{s}^{(k)}\big] ds$$ 
and the function $G_k(t) = E[N_{t}^{(k)}] $ satisfies $G_k(t) \leq \mu t +\kappa \mu^{-\alpha}\int_0^t g(t-s)G_k(s)ds$, for which Gronwall lemma implies $\sup_kG_k(t) \leq \mu t C_t(g)$ for some constant $C_t$ depending on $g,\mu,\alpha$ only as soon as $g$ is locally integrable (see, for instance  Lemma 23(i) in Delattre et al.\ 2016). Letting $k\rightarrow \infty$, we infer by monotone convergence that $E\big[N_t\big] < \infty$ and thus $N_t <\infty$ $P$-almost surely follows. From this, one can observe that the stochastic intensity $\lambda$ of $N$ satisfies the desired equation. The extension to a spatial variable, {\it i.e.} passing from $N(t)$ to $N(t,s)$ and $\lambda(t)$ to $\lambda(t,s)$ satisfying \eqref{rec2} is straighforward.
\end{proof}


{\bf Mean and variance.}

The mean of the recursive process (\ref{rec2}) can be obtained simply by 
using the Georgii-Nguyen-Zessin property of the conditional intensity (Georgii 1976, Nguyen and Zessin 1979).
\begin{eqnarray}
\frac{1}{T} EN(S) &=& \frac{1}{T} E \int \limits _S dN  \nonumber \\
&=& \frac{1}{T} E \int \limits _B \int \limits _0 ^T \lambda_{s,t} \, d\mu(s,t)  \nonumber \\
&=& \frac{1}{T} E \int \limits _B \int \limits _0 ^T 
\{\mu + \kappa \int \limits _B \int \limits _0 ^t \lambda_{s',t'}^{-\alpha} \, g(s-s',t-t') dN_{s',t'}\} d\mu(s,t)  \nonumber \\
&=& \mu + \frac{\kappa}{T} E \int \limits _B \int \limits _0 ^T \int \limits _B \int \limits _0 ^t 
\lambda_{s',t'}^{1-\alpha} \, g(s-s',t-t') d\mu(s,t,s',t')  \nonumber \\
&=& \mu + \frac{\kappa}{T} E \int \limits _B \int \limits _0 ^T \lambda_{s',t'}^{1-\alpha} \, \left 
\{ \int \limits _B \int \limits _0 ^{T-t'} g(s-s',t-t') d\mu(s,t) \right \} d\mu(s',t')  \nonumber \\
&\rightarrow & \mu + \frac{\kappa}{T} E \int \limits _B \int \limits _0 ^T \lambda_{s',t'}^{1-\alpha} \, d\mu(s',t'), \label{meangeneral}
\end{eqnarray}
as $T \rightarrow \infty$,
provided 
\begin{eqnarray}
\lim _{T \rightarrow \infty} \int_B \int \limits _0 ^{T-t'} g(s-s',t-t') ds dt = 1, \,  \forall (s',t'). \label{density}
\end{eqnarray}
If assumption (\ref{density}) is violated then equation (\ref{meangeneral}) is merely an approximation. Impacts of violations to assumption (\ref{density}) are investigated in Schoenberg (2016). 

Note that for the standard recursive model, $\alpha = 1$, and 
(\ref{meangeneral}) reduces simply to 
\begin{eqnarray} \label{recmean}
\mu + \kappa.
\end{eqnarray}
This highlights a major difference between Hawkes models and recursive models. 
For a Hawkes process, doubling the background rate amounts to doubling the total expected number of points, but this is far from true for the recursive process. 
As an example, in the rather realistic simulations in Figure 1a where $\mu = 0.1$ and $\kappa = 2$, 
doubling $\mu$ would only increase the total expected number of points by less than $5\%$, and 
in the case of the process simulated in Figure 1c where $\mu = 0.01$ and $\kappa = 2$, 
doubling $\mu$ would increase the total expected number of points by less than $0.5\%$.\\

{\bf Law of large numbers.}

 We specialise in this section to the case $\alpha=1$ and show that $T^{-1}N_T$ converges to $\mu+\kappa$ as $T \rightarrow \infty$ with rate of convergence $\sqrt{T}$ in $L^2$. For simplicity, we only consider the temporal model $N_t$ with stochastic intensity
 $$\lambda_t=\mu+\kappa \int \limits_0^{t_-}\lambda_s^{-1}g(t-s)dN_s.$$
 
 \begin{prop}
 Assume $\limsup_{T \rightarrow \infty}T^{1/2}\int_T^\infty g(t)dt < \infty$. Then
 $$\sup_TTE\big[\big(T^{-1}N_T-(\mu+\kappa)\big)^2\big]<\infty.$$
 \end{prop}
 \begin{proof}
Write $T^{-1}N_T-(\mu+\kappa) = A_T+B_T$, with
$$A_T = T^{-1}N_T-\frac{1}{T}\int \limits_0^T \lambda_sds\;\;\text{and}\;\;B_T=\frac{1}{T}\int \limits_0^T \lambda_sds-(\mu+\kappa).$$
We claim that both $\sup_TTE[A_T^2]<\infty$ and $\sup_TTE[B_T^2]<\infty$, from which the proposition readily follows.
Let us first consider the term involving $B_T$. We have 
\begin{align*}
 B_T & = \mu+\kappa \frac{1}{T}\int \limits_0^T g(T-s) \int \limits_0^{s-} \frac{dN_u}{\lambda_u} ds - (\mu+\kappa)\\
  & = \kappa \big(\frac{1}{T}\int \limits_0^T g(T-s) \widetilde N_sds - 1\big)\\
  & = \kappa \frac{1}{T}\int \limits_0^T g(T-s) (\widetilde N_s-s)ds + \kappa\big(\frac{1}{T}\int \limits_0^T g(T-s)sds-1\big) \\
  & = B_T^{(1)}+B_T^{(2)},
  \end{align*}
say, where $\widetilde N_s = \int \limits_0^{s-} \frac{dN_u}{\lambda_u} ds$. Clearly
\begin{align*}
(\frac{1}{T}\int \limits_0^T g(T-s)sds-1 & = -\frac{1}{T}\int \limits_0^Tg(s)ds+\int \limits_0^Tg(s)ds-1 \\
& = -\frac{1}{T}\int \limits_0^Tg(s)ds+\int \limits_T^\infty g(s)ds 
\end{align*}
and this (deterministic) term is $O(T^{-1/2})$ by assumption and thus $B_T^{(2)}$ has the right order. As for $B_T^{(1)}$, since $s \mapsto g(T-s)$ is a probability density, we successively use Jensen's inequality, Fubini, the fact that $\widetilde N_s$ is a martingale with predictable compensator $s$, hence $(\widetilde N_s-s)^2$ itself also a martingale with predictable compensator $s$ to obtain
\begin{align*}
E\big[\big(B_T^{(1)}\big)^2\big] & \leq \kappa^2 \frac{1}{T^2}\int \limits_0^T g(T-s) E\big[(\widetilde N_s-s)^2\big]ds \\
& = \kappa^2 \frac{1}{T^2}\int \limits_0^T g(T-s) E\big[\langle \widetilde N_\cdot -\cdot\rangle_s\big]ds \\
& =  \kappa^2 \frac{1}{T^2}\int \limits_0^T g(T-s) sds \\
& = \kappa^2  \frac{1}{T^2}\int \limits_0^T g(T-s) sds
\end{align*}
and this term mutiplied by $T$ is negligible, as for the term $B_T^{(2)}$. We finally turn to the important term $A_T$. Since $N_t-\int \limits_0^t \lambda_sds$ is a martingale, its predictable compensator is also $\int \limits_0^t \lambda_sds$. It follows that
\begin{align*}
E\big[A_T^2\big] & = T^{-2} E\big[\big(N_T-\int_0^T\lambda_sds\big)^2\big] \\
& = T^{-2}E\big[\langle N_\cdot - \int \limits_0^\cdot \lambda_s ds\rangle_T\big] \\
& = T^{-2}\int_0^TE\big[\lambda_s\big]ds.
\end{align*}
The remainder of the proof consist in showing that $\sup_{s>0}E[\lambda_s]<\infty$. This follows in the same line as for the proof of non accumulation of jumps in Proposition \ref{existence}.
 \end{proof}

{\bf Productivity.}

The productivity of a point $\tau_i$ is typically defined in the context of Hawkes or ETAS processes as the expected number of first generation offspring of the point $\tau_i$. For a Hawkes process, the productivity of each point is simply $K$.

In the case of the recursive model (\ref{rec}), the productivity of any point $\tau_i$ is given by $H\{\lambda(\tau_i)\}$. Thus the total productivity, for $n$ points $\tau_1, \tau_2, ..., \tau_n$, is $\sum \limits _{i=1}^n H\{\lambda(\tau_i)\}$, and for the special case 
of the standard recursive model (\ref{oneover}), 
the expected value of the total productivity is 
\begin{eqnarray*}
\kappa E \int \limits _B \int \limits _0 ^T \frac{1}{\lambda_{s,t}} \, dN_{s,t} = \kappa E \int \limits _B \int \limits _0 ^T 
\frac{1}{\lambda_{s,t}} \lambda_{s,t} \, d\mu = \kappa T.
\end{eqnarray*}
Thus under assumption (\ref{density}) the average productivity for the standard recursive model is $\kappa T / N(S) \rightarrow \kappa / (\mu + \kappa)$ a.s., since $N(S)/T \rightarrow \mu + \kappa$ a.s. 
This highlights another difference between the recursive and Hawkes models. 
For a Hawkes process, the points $\tau_1, \tau_2, ...$ all have constant productivity, $K$. For a standard recursive process, 
the productivity of the first point is very large ($1 / \mu$), but the productivity decreases thereafter, ultimately averaging
$\kappa / (\mu + \kappa)$.\\

%
%

{\bf Declustering.} 

In the seismological context, one is often interested in {\sl mainshocks}, and it can occasionally be desirable to remove the earthquakes that could be considered {\sl aftershocks} from a catalog.
Zhuang et al.\ (2002) proposed a method of {\sl stochastic declustering} for Hawkes or ETAS processes whereby one 
may assign to each observed point $\tau_i$ a probability that it was mainshock, attributable to the background rate $\mu$, and 
to each pair of points $(\tau_i, \tau_j)$ one may compute the probability that earthquake $j$ was {\sl triggered} by 
earthquake $i$, and may thus be considered an {\sl aftershock} of event $i$.

Similarly, when discussing the spread of a contagious disease in a given spatial region, one may consider 
the probability that events generated by the recursive model (\ref{rec}) are new outbreak points, attributable to the background rate $\mu$, or whether point $\tau_j$ was infected by point $\tau_i$. Such triggering or infection probabilities would be extremely relevant to a statistical analysis of epidemic data.

Fortunately these background and infection probabilities are very easy to calculate for the recursive model.
Whereas in a sub-critical Hawkes process, the expected proportion of background points is $1/(1-K)$, 
for the standard recursive process, this proportion is $\mu / (\mu + \kappa)$. 
This follows directly from the formula (\ref{recmean}) for the mean of the recursive process.
Referring to the form of the recursive model in (\ref{rec}),
for any points $\tau_i < \tau_j$, the probability that subject $j$ was infected by subject $i$ is given by
\begin{eqnarray}
\frac{H(\lambda_{\tau_i}) g(\tau_j - \tau_i)}{\mu + \int \limits _B \int _0 ^{\tau_j} H(\lambda_{s',t'}) \, g\{\tau_j-(s',t')\} dN_{s',t'}} \label{triggering},
\end{eqnarray}
which can readily be computed. In equation (\ref{triggering}) we are using the simplified notation 
$g(\tau_j - \tau_i)$ to refer to $g(s-s',t-t')$, where $\tau_j = (s,t)$ and $\tau_i = (s',t')$.\\

\section{Simulation.}

One way to simulate a recursive point process is using the thinning technique of Lewis and Shedler (1979).
Specifically, one sets $b$ to some large value,
generates a homogeneous Poisson process of {\sl candidate} points with rate $b$ on the spatial-temporal domain $S$,
sorts the candidate points in order of time, and for each candidate point $\tau_i$, for $i = 1, 2, ...,$,
one keeps the point independently of the others with probability
$\lambda(\tau_i)/b$.
Here, $\lambda(\tau_i)$ is computed using equation (\ref{rec}), where in calculating  
\begin{eqnarray*}
\lambda(s,t) = \mu + \int \limits _B \int \limits _0 ^t H(\lambda_{s',t'}) \, g(s-s',t-t') dN(s',t') = 
\mu + \sum \limits _{i: t_i'<t} H(\lambda_{s',t'}) \, g(s-s',t-t'),
\end{eqnarray*}
the sum is taken over only the {\sl kept} points.
Hawkes processes may be simulated in a similar manner. \\ 

\begin{figure}[hhh]
\includegraphics[width=6in,height=1.2in]{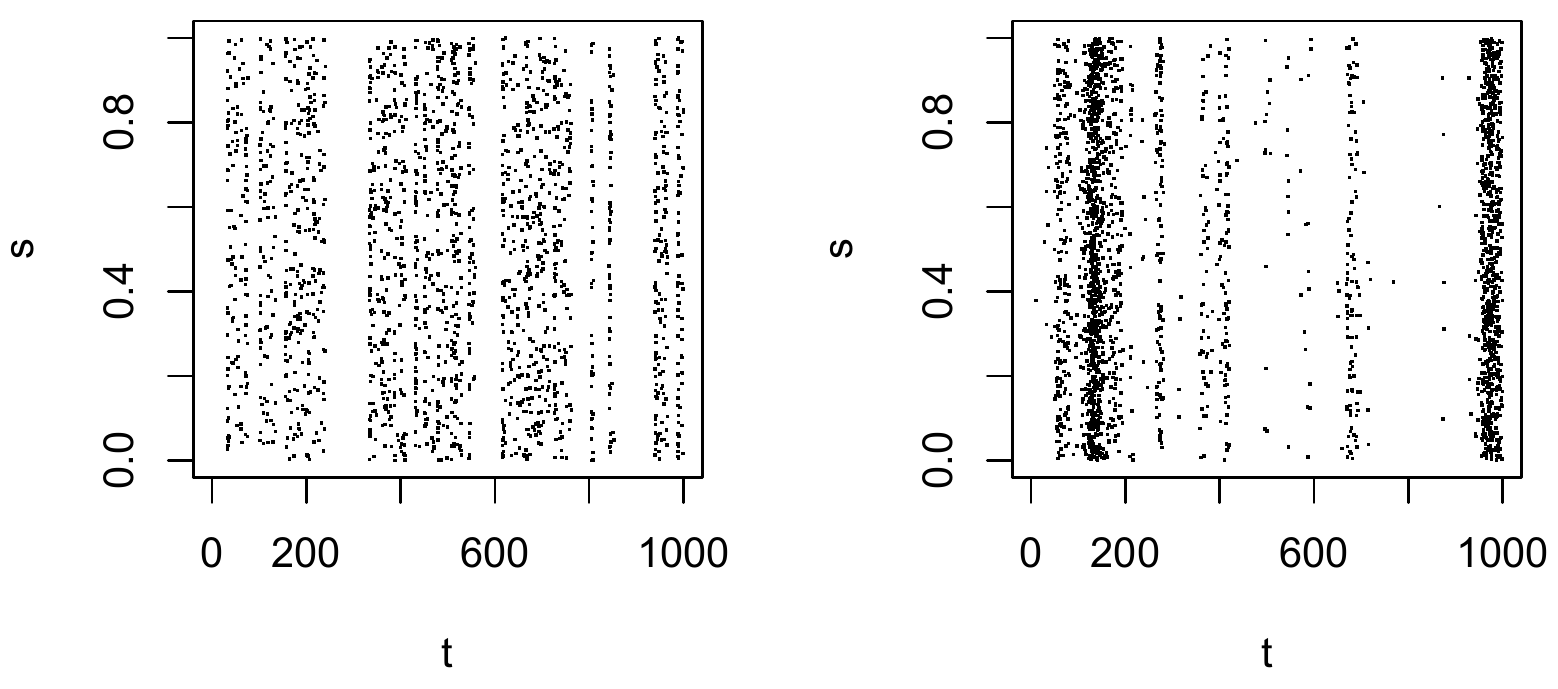}
\includegraphics[width=6in,height=1.2in]{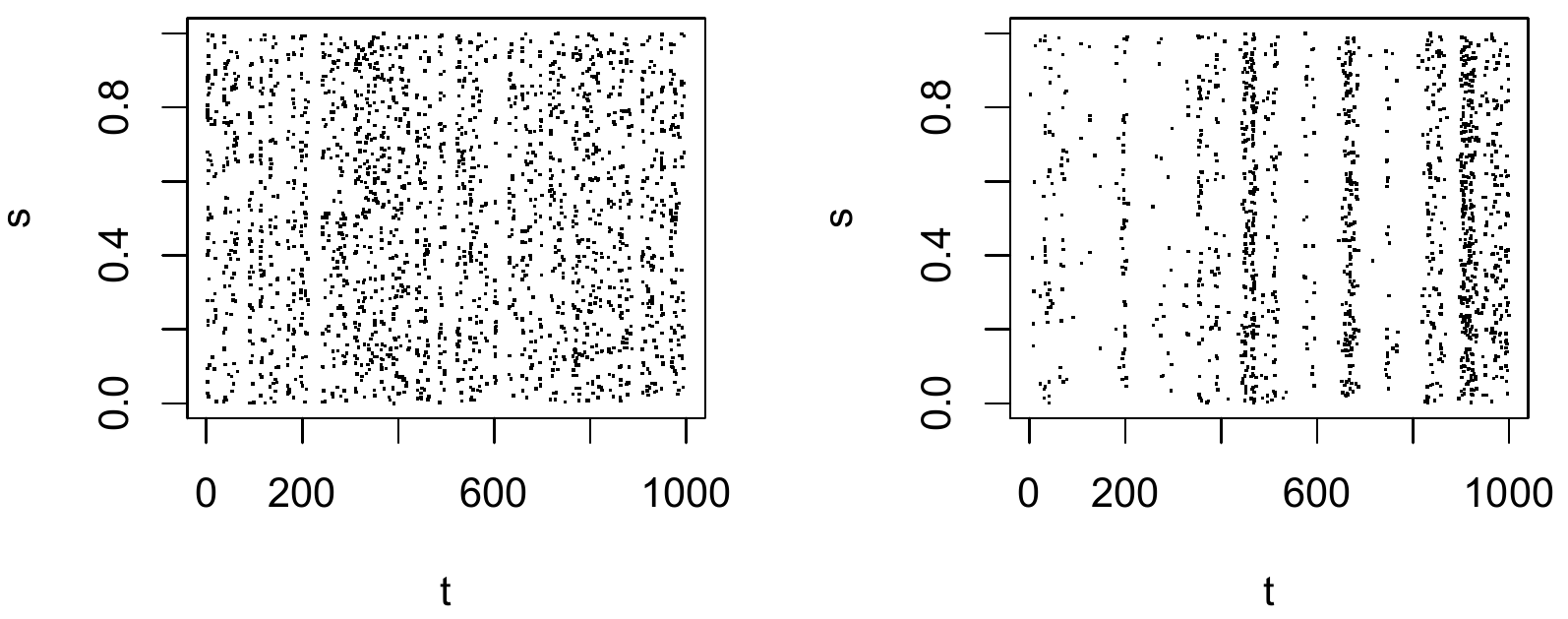}
\caption{(a) Simulation of a standard recursive model (\ref{oneover}) with $\mu = 0.05$, $\kappa = 2$, and $g$ as in (\ref{exponential}) with $\beta_t = 0.8$.
(b) Simulation of a Hawkes model (\ref{hawkes}) with the same $g$ and $\mu$ as in (a), and with $K = \mu/(\mu+\kappa)$ so that the processes in (a) and (b) have the same expected number of points.
(c) Simulation of a standard recursive model (\ref{oneover}) with $\mu = 0.1$, $\kappa = 2$, and $g$ as in (\ref{exponential}) with $\beta_t = 1$.
(d) Simulation of a Hawkes model (\ref{hawkes}) with the same $g$ and $\mu$ as in (c), and with $K = \mu/(\mu+\kappa)$ so that the processes in (c) and (d) have the same expected number of points.
All 4 simulations are over the same domain $S = B \times [0,T]$ with $B = [0,1]$ and $T=1000$. 
For all 4 simulations, points are distributed spatially uniformly in $B$.}
\end{figure}


\begin{figure}[hhh]
\includegraphics[width=6in,height=2in]{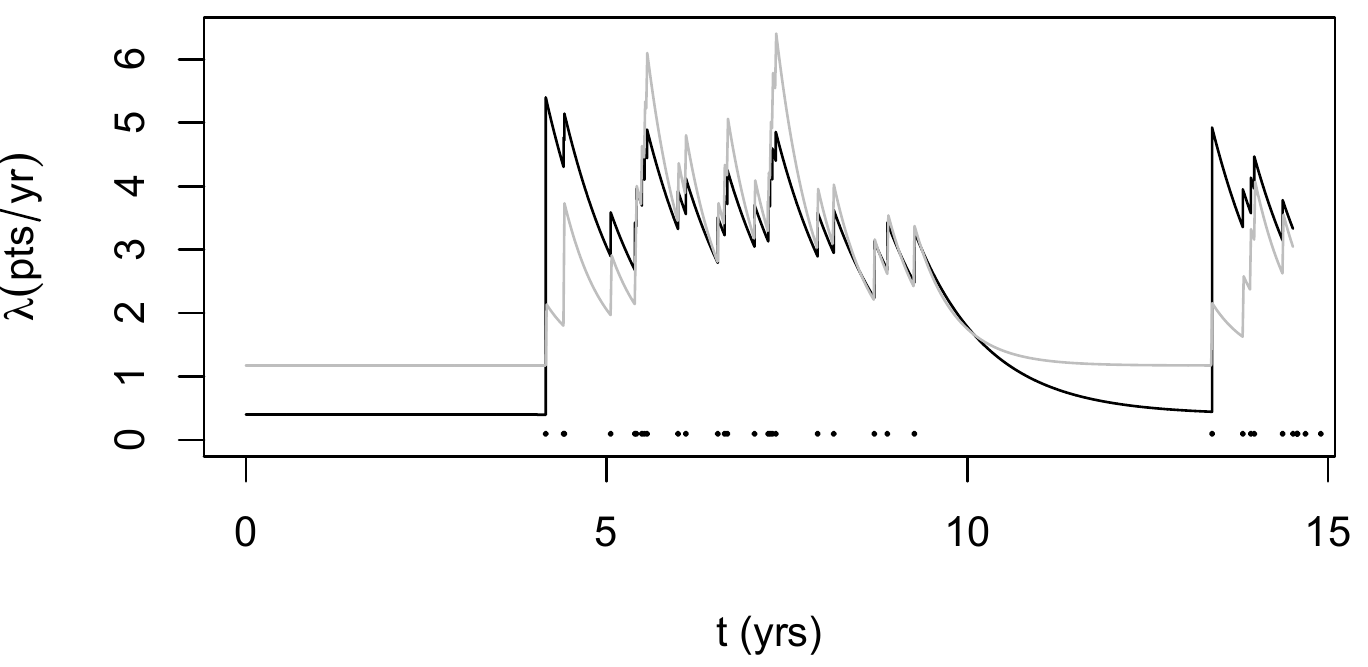}
\caption{The conditional intensity for a simulation of the standard recursive model (black)
and the corresponding conditional intensity for a Hawkes model (grey) fit by maximum likelihood to the same simulated recursive process. Points are shown on the $x$ axis.
The simulated recursive model has $\mu = 0.4$, $\kappa = 2$, $g$ as in (\ref{exponential}) with $\beta_t = 1$, and $T = 1000$, and the Hawkes model also has $g$ as in (\ref{exponential}) and parameters $\mu, K,$ and $\beta_t$ fit by MLE. Points are distributed uniformly over $B = [0,1]$.}
\end{figure}

Figure 1 shows simulations of a recursive process 
and a Hawkes process over the same domain, with the same exponential triggering density, and the same background rate $\mu$. 
In the top panels, $\mu = 0.05$ and $\beta_t = 0.8$, and in the bottom panels, $\mu = 0.1$ and $\beta = 1$. In each case, the parameter $K$ of the Hawkes process was selected as $\kappa / (\mu + \kappa)$ so that the Hawkes and recursive processes would have the same expected number of points. 
One main difference between the Hawkes and recursive models is that the former exhibits occasional small clusters with
just a few or even just one isolated point, whereas the latter produces almost exclusively large clusters. 
One sees also how the parameter $\beta$ influences the degree of clustering in the processes.\\

Figure 2 shows the conditional intensity, $\lambda$, of a standard recursive process (\ref{oneover}), along with the corresponding 
conditional intensity of a Hawkes process (\ref{hawkes}) fit to the simulated recursive process by maximum likelihood.
One sees that the conditional intensity of the recursive process is higher following the initial point in a cluster, but the Hawkes conditional intensity becomes higher after several points have occurred in succession.\\

\section{Estimation.}

As with most space-time point process models including Hawkes and ETAS processes, the parameters in 
recursive point processes can be estimated by maximizing the loglikelihood,
\begin{eqnarray}
\ell(\theta) = \int \limits _S \log \lambda(s) dN(s) - \int \limits_S \lambda(s) d\mu,
\label{loglike}
\end{eqnarray}
where $\theta$ is the vector of parameters to be estimated.
Maximum likelihood estimates (MLEs) of the parameters in such point process models are consistent, asymptotically normal, and efficient (Ogata 1978).

Despite the recursive nature of the model, the loglikelihood of a recursive point process can be computed quite directly. 
For any given realization of points $\{\tau_1, \tau_2, ..., \tau_n\}$, 
given a particular value of the parameter vector $\theta$, 
$\lambda(\tau_1) = \mu$ so one can immediately compute $H\{\lambda(\tau_1)\} = H(\mu)$, 
and thus $\lambda(\tau_2) = \mu + H(\mu)g(\tau_2 - \tau_1)$. 
One therefore has 
$H\{\lambda(\tau_1)\} = H(\mu + H(\mu)g(\tau_2 - \tau_1)$, and one can continue in this fashion 
to compute $\lambda(\tau_i)$ for $i = 1, 2, ..., n$. 

The integral term $\int \limits_S \lambda(s) d\mu$ may readily be approximated in the standard way 
(see e.g.\ Schoenberg 2013).
Assuming $g(t)$ has negligible mass for $t > T-\tau_i$, one may invoke the approximation 
\begin{eqnarray*}
\int \limits _S \lambda(s) d\mu &=& \int \limits _0 ^T \{\mu + \int \limits _0 ^t H(\lambda(s)) g(t-s) dN(s)\} dt\\
&=& \mu T + \int \limits _0 ^T H(\lambda(s)) \int \limits _0 ^{T-s} g(u) dN(s) du\\
& \approx & \mu T + \int \limits _0 ^T H(\lambda(s)) dN(s)\\
& = & \mu T + \sum \limits _i H(\lambda(\tau_i)),
\end{eqnarray*}
which is trivial to compute.
The parameter vector $\theta$ maximizing the approximation of (\ref{loglike}) can then be estimated by standard Newton-Raphson optimization routines. In what follows, we use the function optim() in {\sl R}.
Approximate standard errors can be derived via the diagonal elements of the inverse of the Hessian of the loglikelihood
(Ogata 1978).\\

\section{Application to Measles in Los Angeles, California.}

\begin{figure}[hhh]
\includegraphics[width=6in, height = 2in]{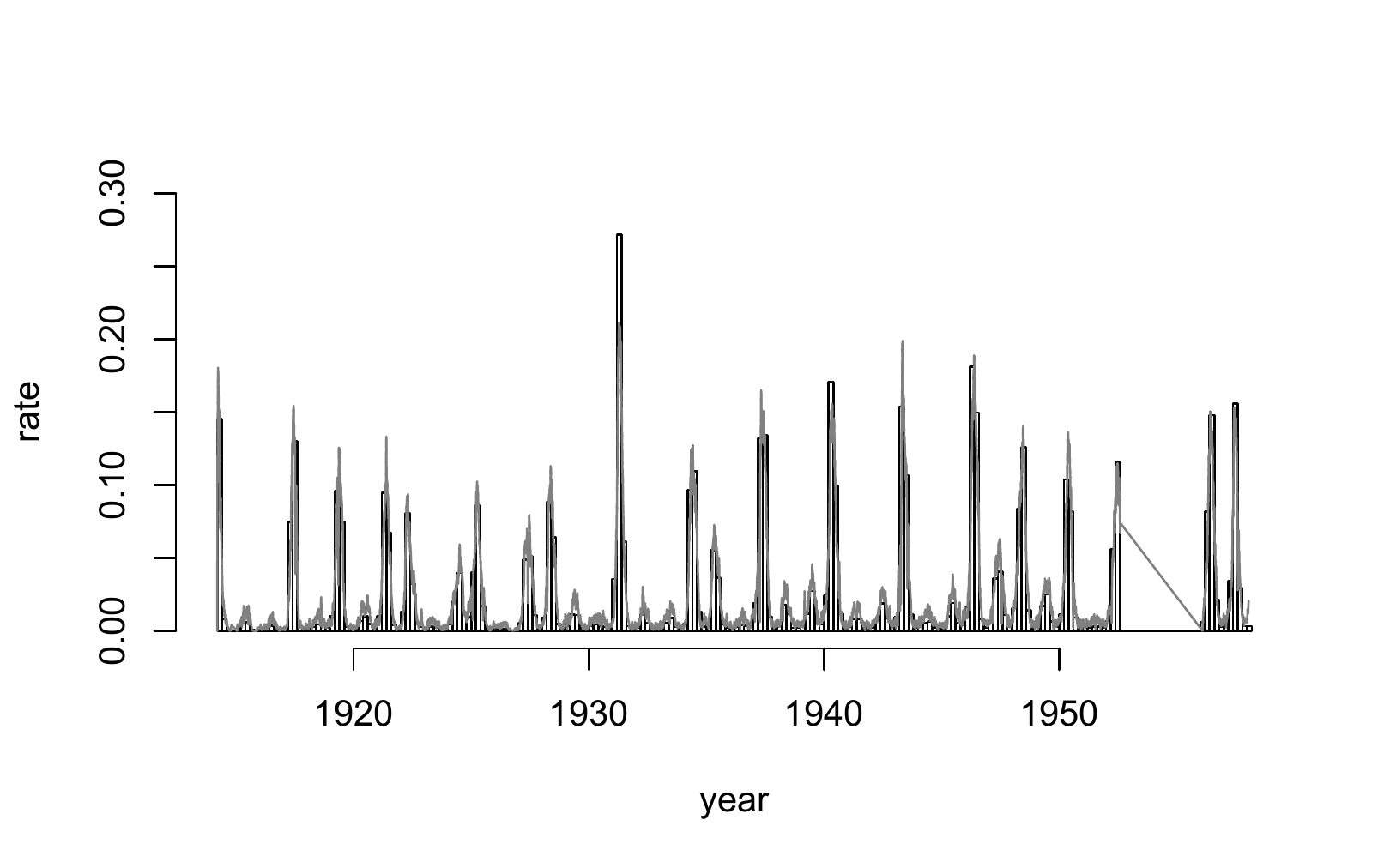}
\caption{Histogram (black bars) of measles cases in Los Angeles, California, from 1910-1956, along with the estimated rate 
of the recursive model (gray curve) with exponential triggering function, fit by maximum likelihood. }
\end{figure}

Recorded cases of measles in Los Angeles, California, from 1/1/1906 to 12/31/1956 were obtained from Project Tycho, 
www.tycho.pitt.edu (Van Panhuis et al., 2013). 
The data consist of weekly counts of confirmed cases of measles in Los Angeles published by the United States Centers of Disease Control (CDC) in its open access weekly Morbidity and Mortality Weekly Reports. For some weeks no information is available, especially in the years 1906-1909, so for this analysis we restrict our attention to the 148,037 recorded cases during the period from 2/5/1910 to 12/31/1956. 
Weeks with no data over this period were treated as having zero confirmed cases.  
Since the temporal resolution of the data is by week, the onset time for each individual case was drawn uniformly within each 
7 day time interval. 
Figure 3 shows a histogram of the cases, along with the estimated rate of the recursive model (\ref{rec2}) with exponential triggering fit to the data by maximum likelihood. 
The estimated parameters are $(\mu, \kappa, \beta, \alpha)$ = (3.907 points/yr, 27.06 triggered points/observed point, 60.01 points/yr, 0.3632), with corresponding standard error estimates $(.1494, .6569, .8681, .06997.)$.

One way to check for convergence of the MLE is to compute the ratio $\int \limits _0 ^T \hat \lambda(t) dt / N(0,T)$, which should be close to 1 since $E \int \limits _0 ^T \hat \lambda(t) dt = E \int \limits _0 ^T dN = N(0,T)$. In practice this ratio often assumes a value extremely close to 1 after fitting by maximum likelihood (Harte 2015). For the measles data, the ratio is 0.995601, and the loglikelihood is 1223183.\\ 

\begin{figure}[hhh]
\includegraphics[width=6in, height = 2in]{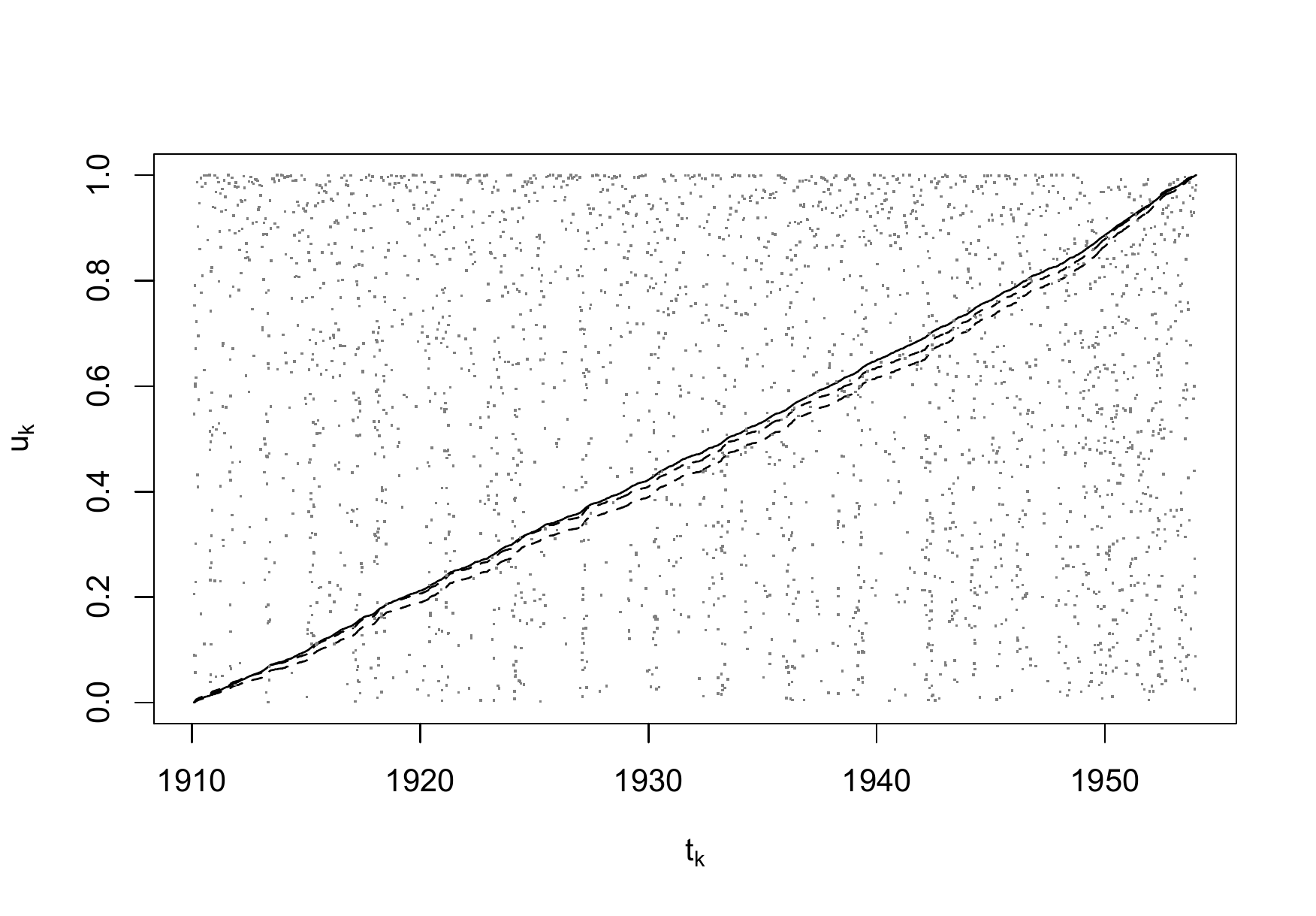}
\caption{Super-thinned residuals $t_k$ using $b = $ 100 points/year and their corresponding standardized interevent times $u_k$. The solid line shows, for each value of $t_k$, the normalized cumulative sum $\sum \limits _{i=1}^{k} u_i / \sum \limits _{i=1}^m u_i$, where $m$ is the number of super-thinned residuals. Dotted lines show lower and upper simultaneous 95\% confidence bounds based on 1000 simulations of the normalized cumulative sums of $m$ uniform random variables. } 
\end{figure}

In order to assess the fit of the model, we used super-thinned residuals (Clements et al.\ 2013). In super-thinning, one selects a constant $b$, thins the observations by keeping each observed point $\tau_i$ independently with probability $b/ \hat \lambda(\tau_i)$ if $\hat \lambda(\tau_i) > b$, 
and superposes points from a Poisson process with rate $(b - \hat \lambda) {\bf 1}_{\hat \lambda \leq b}$, where ${\bf 1}$ denotes the indicator function. The resulting super-thinned residuals form a homogeneous Poisson process with rate $b$ iff.\ $\hat \lambda$ is the true conditional rate of the observed point process (Clements et al.\ 2013). If $t_i$ are the times of the super-thinned points, one may consider the interevent times,
$r_i = t_{i} - t_{i-1}$ (with the convention $t_0 = 0$), which should be exponential with mean $1/b$ if the fitted model $\hat \lambda$ is correct, and it is natural therefore to inspect the uniformity of the standardized interevent times $u_i = F^{-1}(r_i)$, where $F$ is the cumulative distribution function of the exponential with mean $1/b$. 
Figure 4 shows the super-thinned residuals $t_i$ along with their corresponding standardized interevent times $u_i$, as well as the cumulative sum of the standardized intervent times, along with individual 95$\%$ confidence bounds based on 1000 simulations of an equivalent number of uniform random variables. The super-thinned residuals appear to be well scattered, though the model does not fit perfectly; there are more very large interevent times than expected, especially between 1925 and 1945, and the cumulative sum of the standardized interevent times is somewhat more concave than expected as a result. 
These largest interevent times appear to be somewhat clustered together, while the other interevent times appear to be largely well scattered, 
as shown in the lag plot of the standardized interevent times in Figure 5. \\

\begin{figure}[hhh]
\includegraphics[width=6in, height = 2in]{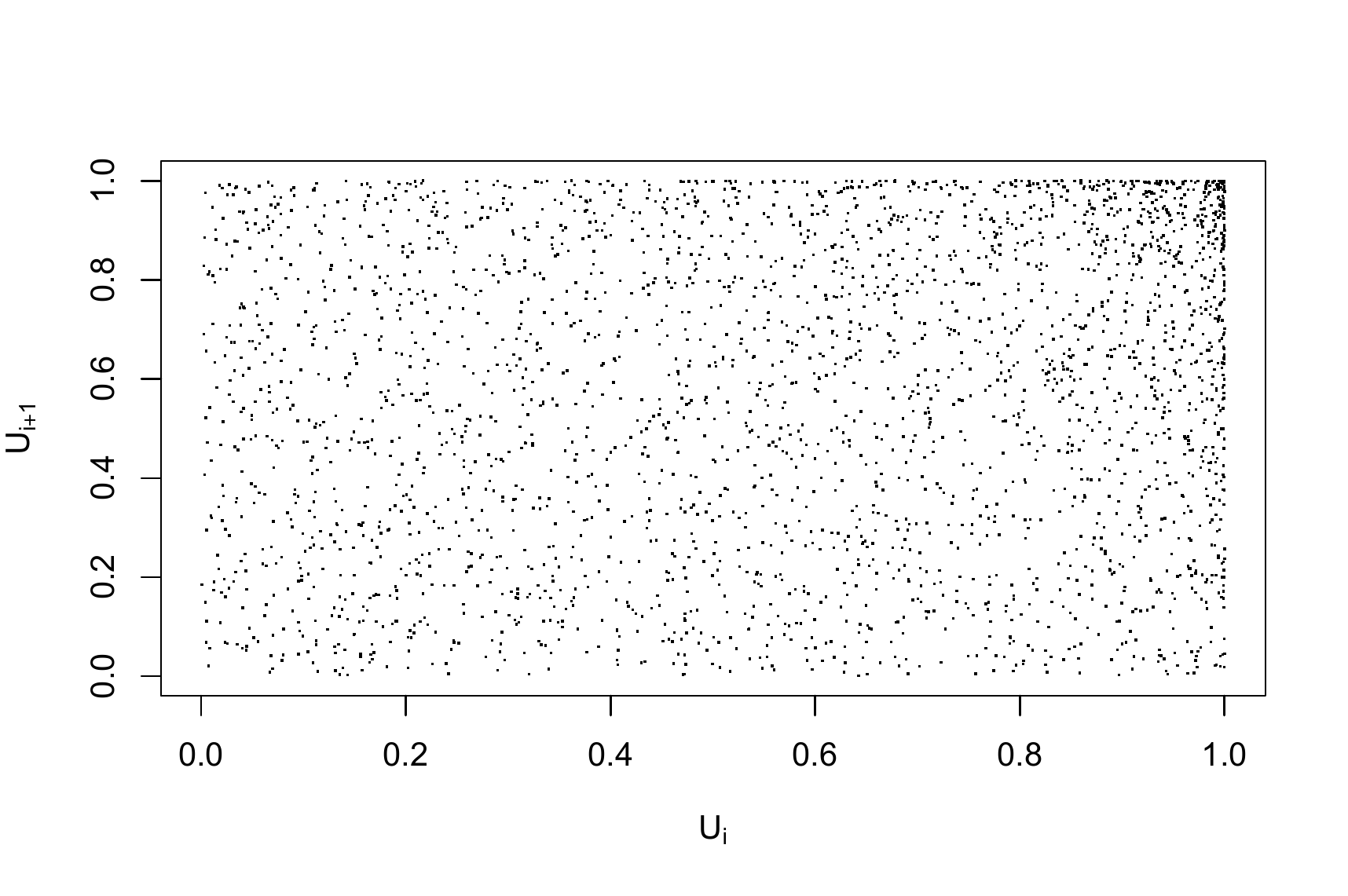}
\caption{Lag plot of the standardized interevent times $u_i$ of the super-thinned residuals using $b = $ 100 points/year. } 
\end{figure}

\begin{figure}[hhh]
\includegraphics[width=6in, height = 2in]{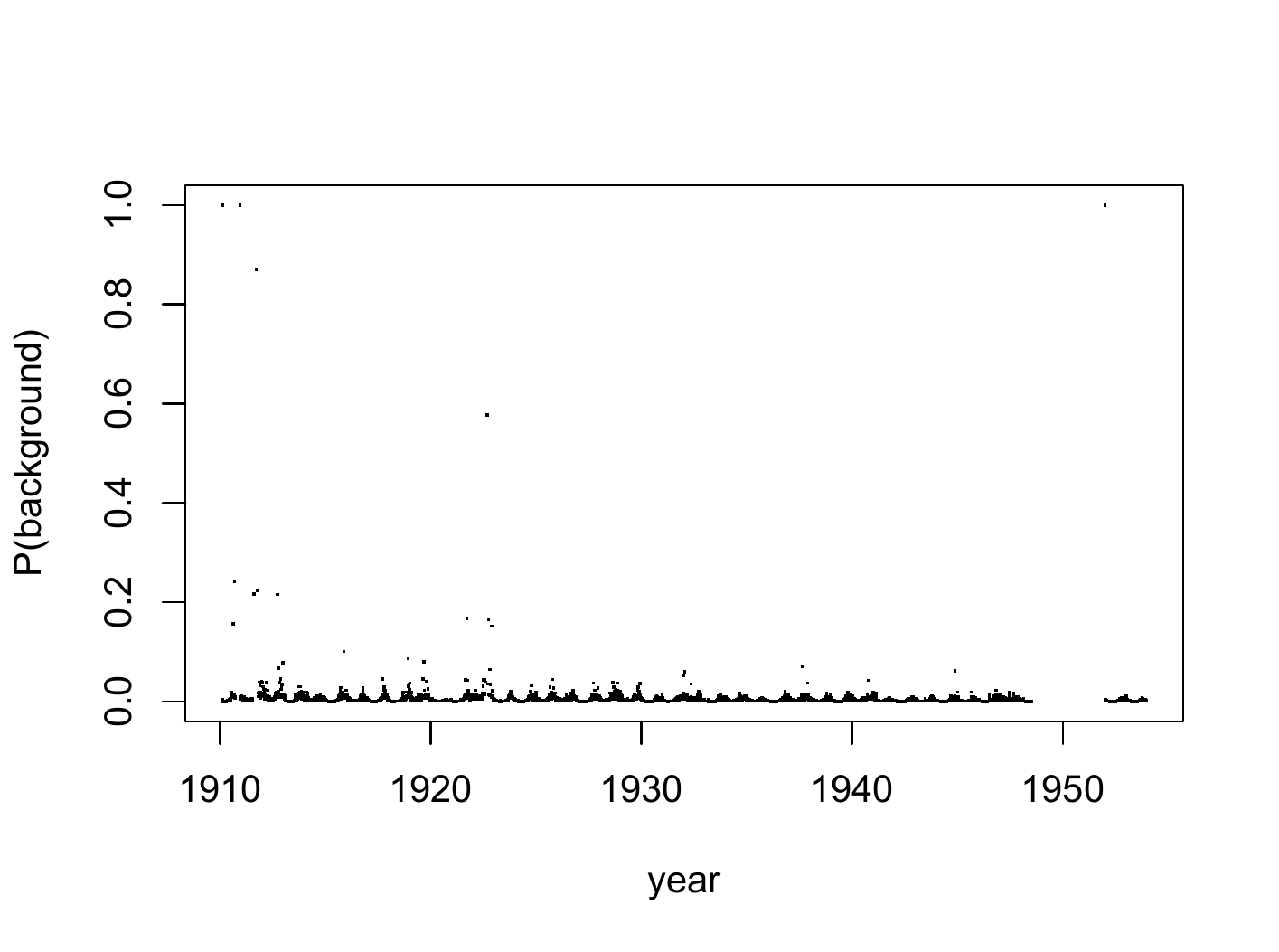}
\caption{Stochastic declustering of the Los Angeles measles cases based on the fitted model (\ref{rec2}). For each observed point $t_i$, the $y$-coordinate, $\mu / \lambda(t_i)$, is the probability, based on model (\ref{rec2}), that the point is attributed to the background rate ($\mu$) as opposed to contagion from previous points. } 
\end{figure}

Figure 6 shows the stochastic declustering of the measles data in Los Angeles using the fitted model (\ref{rec2}). 
The $y$-axis shows the probability, based on the fitted model (\ref{rec2}), that the point is attributed to the background rate ($\mu$) as opposed to contagion from previous points.
The vast majority of points are attributed to contagion rather than novel outbreaks. Certain particular points in 1910-1912, 1922, and 1952 are given much higher likelihood of being attributed to new outbreaks, though this is likely at least partially an artifact of missing data in preceding weeks.

\section{Concluding remarks.}
 
The recursive point process model proposed here and fit to infectious disease data seems to be an improvement over Hawkes models because of its more general form and its flexibility, enabling it to account for changes in the rate of contagion over the course of an epidemic. 
We should note that although Hawkes models are widely used in seismology and are occasionally called {\sl epidemic-type} models, and although the processes by which humans spread contagious diseases seem naturally to lend themselves to such models, the use of Hawkes models in describing the spread of infections has been scant. 
Exceptions are Becker (1977), who proposed purely temporal self-exciting point process models to describe the temporal spread of smallpox in Brazil, Farrington et al.\ (2003), who describe the effect of vaccinations on the spread of measles in the United States using self-exciting point process models, and Balderama et al.\ (2012), who model invasive red banana plant locations and times
using a parametric space-time Hawkes point process model. 
As noted by Law et al.\ (2009), unlike grid-based studies on area occupation, where the surface of study is divided into an array of pixels on a grid, spatial-temporal point processes can enable greater precision of forecasts in space and time, and can offer a more detailed and precise account of spatial heterogeneity and clustering.
Diggle (2006) investigated inhomogeneity in foot-and-mouth disease using spatial-temporal point process models estimated by partial likelihood methods, and Diggle (2014) discusses some successful uses of spatial-temporal point process modeling in describing in detail ecological phenomena such as the locations of Japanese black pine saplings as well as public health data such as liver cirrhosis in Northeastern England, but these efforts currently do not appear to have been widely replicated. 
Perhaps the added flexibility of the recursive model proposed here will facilitate the more frequent use of point process models for such epidemic data.

According to the fitted recursive model for measles in Los Angeles from 1910 to 1956, the vast majority of observed points were spread via contagion, with only a small fraction (0.121\% of cases, or 3.907 cases per year) attributable to new outbreaks. 
The fitted exponential triggering function in the recursive model had an estimated rate of 60.01 points/year, which corresponds to a mean triggering time of 6.08 days for each transmission. 
This estimate of ~6 days for each transmission is, given the epidemiology of the measles virus, plausible. For instance, CDC reports that communicability of measles can occur from 4 days before, to 4 days after, onset of symptomatic rash, and that rashes present on patients between 7-21 days after exposure (Centers for Disease Control and Prevention, 2015). 
Note that this estimate of contagion is based on when the measles cases were {\sl reported}; thus, the contagion suggested here within 6.08 days corresponds to cases being reported within 6.08 days of one another. This may differ from the actual delay times between subjects' contraction of the disease.
There may be numerous covariates, such as climate, geographical and geological variables for instance, that are omitted here yet may influence the relationship observed here between previously observed points and the rate of future points. 
The conditional intensity may nevertheless be consistently estimated in the absence of such information provided the impact of the missing covariates is suitable small, as shown in Schoenberg (2016). 

We have presented an extension of the Hawkes point process model, a recursive model, that allows for previous disease status to inform a flexible component describing the time intervals between contagious events. In the special case where the productivity is inversely proportional to the conditional intensity (i.e. when $\alpha = 0$), we have shown that this standard recursive model is computationally trivial to estimate, and does not require estimates of more complex parameters typically needed for accurate estimations of transmission events. We have demonstrated that these recursive models perform well on historical disease datasets, and can lead to insights into the transmission dynamics of particularly contagious diseases. These advances are particularly relevant given the recent outbreaks of such diseases in the same regions tested here, and will hopefully encourage informed strategies as to how best prevent and mitigate future outbreaks.




\section*{Acknowledgements.}
This material is based upon work supported by the National Science Foundation
under grant number DMS 1513657. 

Computations were performed in {\sl R} (www.r-project.org). Thanks to the U.S. CDC for supplying the data and to Project Tycho for making it so easily available. 

Thanks to UCLA and Paris for allowing Professor Schoenberg a one year sabbatical during which time this research was performed.

\pagebreak

\noindent {\sc \bf References}
\vskip 0.1in \parskip 1pt \parindent=1mm \def\reference{\hangindent=22pt\hangafter=1}

\reference






\reference
Bacry, E., Delattre, S., Hoffmann, M. and Muzy, J. F. (2013). Modeling microstructure noise with mutually exciting point processes. 
{\sl Quantitative Finance} {\bf 13}, 65-77.

\reference
Bacry, E., Mastromatteo, I., Muzy, J-F. (2015).
Hawkes processes in finance.
Market Microstructure and Liquidity Vol. 01, No. 01, 1550005.



\reference
Balderama, E., Schoenberg, F.P., Murray, E., and Rundel, P. W., 2012.
Application of branching models in the study of invasive species.
{\sl Journal of the American Statistical Association} 107(498), 467--476.


\reference
Becker, N., 1977.
Estimation for discrete time branching processes with application to epidemics. 
{\sl Biometrics} 33(3), 515-522.







\reference 
Centers for Disease Control and Prevention, 2015. 
{\sl Epidemiology and Prevention of Vaccine-Preventable Diseases}, 13th ed. 
Hamborsky J, Kroger A, Wolfe S, eds. Washington D.C., Public Health Foundation, Chapter 13, pp.\ 209-230.

\reference
Clements, R.A., Schoenberg, F.P., and Veen, A., 2013. 
Evaluation of space-time point process models using super-thinning. 
{\sl Environmetrics} {\bf 23}(7), 606--616. 



\reference
Daley, D., and Vere-Jones, D., 2003.
{\sl An Introduction to the Theory of Point Processes, Volume 1: Elementary Theory and Methods, 2nd ed.},
Springer: New York.

\reference
Daley, D., and Vere-Jones, D. 2007.
{\sl An Introduction to the Theory of Point Processes, Volume 2: General Theory and Structure, 2nd ed.}
Springer: New York.


\reference
Delattre, S., Fournier, N. and Hoffmann, M. 2016.
Hawkes processes on large networks.
{\sl Annals of Applied Probability} 26, 216--261. 

\reference
Diggle, P.J., 2006.
Spatio-temporal point processes, partial likelihood, foot-and-mouth.
{\sl Statistical Methods in Medical Research} 15, 325--336.

\reference 
Diggle, P.J., 2014.
{\sl Statistical Analysis of Spatial and Spatio-temporal Point Patterns,} 3rd ed.
CRC Press, Boca Raton.



\reference
Farrington, C.P., Kanaan, M.N., and Gay, N.J., 2003.
Branching process models for surveillance of infectious diseases controlled by mass vaccination.
{\sl Biostatistics} 4(2), 279--295.



\reference
Georgii, H.O. (1976). 
Canonical and grand canonical Gibbs states for
continuum systems. 
{\sl Communications of Mathematical Physics} 48, 31?51.



\reference
Harte, D.S., 2014. 
An ETAS model with varying productivity rates. 
{\sl Geophysical Journal International} {\bf 198}(1), 270-284. 

\reference
Harte, D.S., 2015.
Log-likelihood of earthquake models: evaluation of models and forecasts.
{\sl Geophysical Journal International} {\bf 201} (2), 711-723.

\reference
Hawkes, A. G., 1971.
Point spectra of some mutually exciting point processes,
{\sl J. Roy. Statist. Soc.}, {\bf B33}, 438-443.

\reference
Law, R., Illian, J., Burslem, D.F.R.P., Gratzer, G., Gunatilleke, C.V.S., and Gunatilleke, I.A.U.N., 2009.
Ecological information from spatial patterns of plants: insights from point process theory.
{\sl Journal of Ecology}, 97(4), 616--628.


\reference
Lewis, P.A.W., and Shedler, G.S. (1979).
Simulation of non-homogeneous Poisson processes by thinning.
{\sl Naval Res. Logistics Quart.} 26(3), 403-413.



\reference
Marsan, D., and Lenglin{\'e}, O., 2008.
Extending earthquakes' reach through cascading.
{\sl Science} 319(5866), 1076--1079.





\reference
Nguyen, X. and Zessin, H. (1979). 
Integral and differential characterizations
of Gibbs processes. 
{\sl Mathematische Nachrichten} 88, 105-115.

\reference
Ogata, Y., 1978. 
The asymptotic behavious of maximum likelihood estimators for stationary point processes.
{\sl Annals of the Institute of Statistical Mathematics}, 30, 243--261.

\reference
Ogata, Y., 1988.
Statistical models for earthquake occurrence and residual
analysis for point processes,
{\sl J.\ Amer.\ Statist.\ Assoc.}, {\bf 83}, 9-27.

\reference
Ogata, Y., 1998.
Space-time point-process models for earthquake occurrences,
{\sl Ann.\ Inst.\ Statist.\ Math.}, {\bf 50}(2), 379-402.

\reference
Ogata, Y., 2004. 
Space-time model for regional seismicity and detection of
crustal stress changes, 
{\sl J. Geophys. Res.} {\bf 109}(B3), B03308, 1-16.

\reference
Ogata, Y., Katsura, K., and Tanemura, M., 2003.
Modelling heterogeneous spacetime
occurrences of earthquakes and its residual analysis, 
{\sl Applied Statistics} {\bf 52}, 499-509.








\reference 
Schoenberg, F.P. 2006. 
On non-simple marked point processes. 
{\sl Ann. Inst. Stat. Math.} {\bf 58}(2), 223-233.

\reference
Schoenberg, F.P., 2013. 
 Facilitated estimation of ETAS. 
 {\sl Bull.\ Seism.\ Soc.\ Amer.} 103(1), 601--605. 

\reference 
Schoenberg, F.P., 2016.
A note on the consistent estimation of spatial-temporal point process parameters. 
{\sl Statistica Sinica} 26, 861-789.

\reference
van Panhuis, W.G., Grefenstette, J., Jung, S.Y., Chok, N.S., Cross, A., Eng, H., Lee, B.Y., Zadorozhny, V., Brown, S., Cummings, D., and Burke, D.S. (2013). 
Contagious diseases in the United States from 1888 to the present. 
{\sl NEJM} 369(22), 2152-2158.

\end{document}